[1994/12/01]
\documentclass[a4paper]{article}
\usepackage{url}
\usepackage{enumerate}
\usepackage{amssymb}
\usepackage{amsmath,amsthm}
\usepackage[left=1in, right=1in, top=1in, bottom=1in]{geometry}
\usepackage[affil-it]{authblk}
\usepackage{MnSymbol}
\usepackage[Symbol]{upgreek}

\title{On the Well Extension of Partial Well Orderings}
\author{Haoxiang Lin}

\usepackage{amsmath,amsthm}

\chardef\bslash=`\\ 

\theoremstyle{definition}

\theoremstyle{remark}

\newcommand{\eval}[2][\right]{\relax
  \ifx#1\right\relax \left.\fi#2#1\rvert}

\theoremstyle{definition}
\newtheorem{theorem}{Theorem}[section]
\newtheorem{lemma}[theorem]{Lemma}

\theoremstyle{definition}
\newtheorem{definition}[theorem]{Definition}

\theoremstyle{remark}

\numberwithin{equation}{section}

\begin{document}
\date{}
\maketitle

\newcommand{\relarank}{\text{RK}}
\newcommand{\revrelarank}{\textit{M}}
\newcommand{\getwell}{\text{GW}}
\newcommand{\getleast}{\text{GL}}
\newcommand{\extendwell}{\text{F}}
\newcommand{\increasingsequence}{\text{H}}
\newcommand{\range}{\text{ran$\,$}}
\newcommand{\field}{\text{fld$\,$}}
\newcommand{\getchoicea}{\text{G}_{1}}
\newcommand{\getdescendant}{\text{GD}}
\newcommand{\extendlinear}{\text{J}}
\newcommand{\descendingchain}{\textit{g}}
\newcommand{\rangerelarank}{\lambda}

\begin{abstract}\label{sec:abs}
In this paper, we study the well extension of strict(irreflective) partial well orderings.
We first prove that any partially well-ordered structure $\langle A, R \rangle$ can be extended to a well-ordered one.
Then we prove that every linear extension of $\langle A, R \rangle$ is well-ordered
if and only if $A$ has no infinite totally unordered subset under $R$.
\end{abstract}

\section{Introduction}\label{sec:intro}

The partial well ordering is a partial ordering which additionally reveals element minimality.
Such a concept is the natural extension of well ordering.
In the study of partial orderings, we first choose either strict(irreflective) or non-strict(reflective) orderings as the basis.
For the non-strict case, we no longer need to specify the set on which the partial ordering is defined.
This is because whenever $R$ is a partial ordering defined on a set $A$, then $A = {\field} R$.
Strict partial orderings lose this advantage, however the whole class of partial orderings is significantly enlarged.

By Order-Extension Principle~\cite{Szpilrajn30}, any partial ordering can be linearly extended.
Similarly, E. S. Wolk proved that
\textit{a non-strict partial ordering $R$ defined on $A$ is a non-strict partial well ordering iff
	every linear extension of $R$ is a well ordering of $A$}~\cite{Wolk1967}.
However, this result does not apply to strict partial well orderings any more.
Take $\langle \mathbb{Z}, \varnothing \rangle$ as an example in which $\mathbb{Z}$ is the set of integers.
Let $<_{\mathbb{Z}}$ be the normal ordering of $\mathbb{Z}$.
Clearly $\varnothing$ is a strict partial well ordering(refer to later definition \ref{def_partial_well}),
however $<_{\mathbb{Z}}$ is a linear extension of $\varnothing$ but not a well ordering.
The reason is that $\varnothing$ is no a legal non-strict partial well ordering at all.

In this paper, we study the well extension of strict partial well orderings
which are largely ignored by previous research work
(~\cite{Higman1952}, ~\cite{Kruskal1960}, ~\cite{Michael1960}, ~\cite{Nash1963}, ~\cite{Nash1964}, ~\cite{Rado1954}, ~\cite{Tarkowski1960}, ~\cite{Wolk1967}).
In the sequel, when we talk about partial or partial well orderings without special emphasis,
we assume that they are strict.
First we show the result that any partially well-ordered structure $\langle A, R \rangle$ can be well extended.
Such a result also applies to a well-founded structure because the well-founded relation can be easily extended to a partial well ordering.
Then we prove that every linear extension of $\langle A, R \rangle$ is well-ordered
if and only if $A$ has no infinite totally unordered subset under $R$.

Given a structure $\langle A, R \rangle$ where $R$ is a binary relation on $A$, we define the following notions:

\begin{definition}
$t \in A$ is said to be an \textit{R-minimal} element of A iff there is no $x \in A$ for which $x\,R\,t$.
\end{definition}

\begin{definition}\label{def_well_founded}
$R$ is said to be \textit{well founded} iff every nonempty subset of $A$ has an $R$-minimal element.
\end{definition}

\begin{definition}\label{def_partial_well}
$R$ is called a \textit{partial well ordering} if it is a transitive well-founded relation.
\end{definition}

A partial well ordering by the above definition \ref{def_partial_well} is strict
because any well-founded relation is irreflexive
otherwise if $x\,R\,x$ then the set $\{x\}$ has no $R$-minimal element.

The following lemma is well known, and we therefore omit its proof.
\begin{lemma}\label{partial_well_property}
	The following properties of a partially ordered structure $\langle A, R \rangle$ are equivalent.
	\leavevmode
	\begin{enumerate}[(a)]
		\item
		$\langle A, R \rangle$ is partially well-ordered.
		
		\item
		There is no function $f$ with domain $\upomega$ and range $A$ such that $f(n^{+}) \, R \, f(n)$ for each $n \in \upomega$
		($f$ or the sequence $\langle f(0), f(1), \cdots, f(n), \cdots \rangle$ is sometimes called a \textit{descending chain}).
	\end{enumerate}
\end{lemma}

We say that two elements $x$ and $y$ are \textit{incomparable}
if and only if $x \neq y$, $\neg (x R y)$  and $\neg (y R x)$.
A subset $B$ of $A$ is \textit{totally unordered} if and only if any two distinct elements of $B$ are incomparable.
To be noted, $A$ can have any arbitrarily large totally unordered subset.
This is a fundamental difference from those non-strict partial well orderings
in that only finite totally unordered subsets exist.
Clearly if $B \not \subseteq {\field} R$, then any $t$ in $B - {\field} R$ is an $R$-minimal element.

\section{${\revrelarank}$-decomposition}\label{sec:decomposition}

We construct a useful canonical decomposition of $A$
by elements' \textit{relative ranks} under $R$ using transfinite recursion.
Such decomposition helps in later proofs.

To be more precise, let $R$-rank be denoted as {\relarank},
then {\relarank} is a function for which ${\relarank}(t) = \{{\relarank}(x) \ | \ x\,R\,t\}$.
{\relarank} is defined by the transfinite recursion theorem schema on well-founded structures.
Take $\gamma_{1}(f, t, z)$ to be the formula $z$ = ran$\ f$.
If $\gamma_{1}(f, y_{1})$ and $\gamma_{1}(f, y_{2})$,
it is obvious that $y_{1} = y_{2}$.
Then there exists a unique function {\relarank} on $A$ for which
\begin{align*}
{\relarank}(t) &= {\range}({\relarank} \upharpoonright \{x \in A \ | \ x\,R\,t\})\\
&= {\relarank} \lsem\{x \ | \ x\,R\,t\} \rsem\\
&= \{{\relarank}(x) \ | \ x\,R\,t\}
\end{align*}

{\relarank} is similar to the "$\epsilon$-image" of well-ordered structures, and has the following properties:
\begin{lemma}\label{relarank_property}
	\leavevmode
	\begin{enumerate}[(a)]
		\item
		For any $x$ and $y$ in $A$,
		\begin{align*}
		x\,R\,y \ &\Rightarrow \ {\relarank}(x) \in {\relarank}(y)\\
		{\relarank}(x) \in {\relarank}(y) \ &\Rightarrow \ \exists z \in A \text{ with } {\relarank}(z) = {\relarank}(x) \text{ and } z\,R\,y
		\end{align*}
		
		\item
		${\relarank}(t) \notin {\relarank}(t)$ for any $t \in A$.
		
		\item
		${\relarank}(t)$ is an ordinal for any $t \in A$.
		
		\item
		{\range}{\relarank} is an ordinal.
	\end{enumerate}
\end{lemma}
\begin{proof}
	\leavevmode
	
	\begin{enumerate}[(a)]
		\item
		By definition.
		
		\item
		Let $S$ be the set of counterexamples:
		\begin{align*}
		S = \{t \in A \ | \ {\relarank}(t) \in {\relarank}(t)\}
		\end{align*}
		If $S$ is nonempty, it has a minimal $\hat{t}$ under $R$.
		Since ${\relarank}(\hat{t}) \in {\relarank}(\hat{t})$, there is some $x\,R\,\hat{t}$ with ${\relarank}(x) = {\relarank}(\hat{t})$
		by (a).
		But then ${\relarank}(x) \in {\relarank}(x)$ and $x \in S$, contradicting the fact that $\hat{t}$ is minimal in $S$.
		
		\item
		Let
		\begin{align*}
		B = \{t \in A \ | \ {\relarank}(t) \text{ is an ordinal}\}
		\end{align*}
		We use Transfinite Induction Principle to prove that $B = A$.
		For a minimal element $\hat{t} \in A$ under $R$, ${\relarank}(\hat{t}) = \varnothing$ which is an ordinal.
		So $\hat{t} \in B$, and $B$ is not empty.
		Assume seg $t = \{x \in A \ | \ x\,R\,t\} \subseteq B$,
		then ${\relarank}(t) = \{{\relarank}(x) \ | \ x\,R\,t \}$ is a set of ordinals by assumption.
		If $u \in v \in {\relarank}(t)$, there exist $y, z$ in $A$
		with $u = {\relarank}(y), v = {\relarank}(z), y\,R\,z$ and $z\,R\,t$.
		Because $R$ is a transitive relation, then $z\,R\,t$ and $u \in {\relarank}(t)$.
		${\relarank}(t)$ is a transitive set of ordinals,
		which implies that it is an ordinal and $t \in B$.
		
		\item
		If $u \in {\relarank}(t) \in {\range}{\relarank}$, then there is some $x\,R\,t$ with $u = {\relarank}(x)$;
		consequently $u \in {\range}{\relarank}$.
		
		Then {\range}{\relarank} is a transitive set of ordinals, therefore itself is an ordinal too.
	\end{enumerate}
\end{proof}

In the sequel, {\range}{\relarank} will be denoted as ${\rangerelarank}$.
To be noted, {\relarank} is not a homomorphism of $A$ onto ${\rangerelarank}$.
We next define
\begin{align*}
{\revrelarank} = \{\langle \alpha, B \rangle \ | \ (\alpha \in {\rangerelarank}) \land (B \subseteq A) \land (x \in B \Leftrightarrow {\relarank}(x) = \alpha)\}
\end{align*}
${\revrelarank}$ is a function from ${\rangerelarank}$ into $\mathcal{P}(A)$,
because it is a subset of ${\rangerelarank} \times \mathcal{P}(A)$ and is single rooted.
Let ${\revrelarank}_{\alpha} = {\revrelarank}(\alpha)$ for $\alpha \in {\rangerelarank}$,
then it is not hard to confirm that ${\revrelarank}_{\alpha}$ is a non-empty set
and ${\revrelarank}  \lsem {\rangerelarank} \rsem = \{ {\revrelarank}_{\alpha} \ | \ \alpha \in {\rangerelarank} \}$
is a partition of set $A$ which will be referred to as the ${\revrelarank}$-\textit{decomposition}.
By lemma \ref{relarank_property},
each ${\revrelarank}_{\alpha}$ is a totally unordered subset of $A$ under $R$.

\section{Well Extension}\label{sec:proof}

In this section, we prove that:
\begin{theorem}\label{partial_well_to_well}
	Any partially well-ordered structure $\langle A, R \rangle$ can be extended to a well-ordered structure $\langle A, W \rangle$ in which $R \subseteq W$.
\end{theorem}

Actually Theorem \ref{partial_well_to_well} also applies to a well-founded structure
because the well-founded relation can be first extended to a partial well ordering:
\begin{lemma}\label{well_founded_to_partial_well}
If $\langle A, R \rangle$ is a well-founded structure, then $R$ can be extended to a partial well ordering on $A$.
\end{lemma}
\begin{proof}
$R$'s transitive extension $R^t$ is a partial well ordering.
Please refer to \cite{Enderton77} for details of this well-known result.
\end{proof}

Clearly if either $A = \varnothing$ or $R = \varnothing$, the extension is trivial by Well-Ordering Theorem.
We assume that both $A$ and $R$ are not empty.
The idea is to linearly extend elements of $A$ from different ${\revrelarank}_{\alpha}$ in ascending order,
and then well extend those in the same ${\revrelarank}_{\alpha}$:
\begin{enumerate}
	\item
	Suppose $x \in {\revrelarank}_{\alpha}, y \in {\revrelarank}_{\beta} \text{ and } x \neq y$.
	\item
	if $\alpha \in \beta$, add $\langle x, y \rangle$ to $W$.
	
	\item
	if $\alpha \ni \beta$, add $\langle y, x \rangle$ to $W$.
	
	\item
	if $\alpha = \beta$, then $x$ and $y$ are incomparable.
	By Well-Ordering Theorem, there exists a well ordering $\prec_{{\revrelarank}_{\alpha}}$ on the set ${\revrelarank}_{\alpha}$,
	and add either $\langle x, y \rangle$ to $W$ if $x \prec_{{\revrelarank}_{\alpha}} y$,
	or $\langle y, x \rangle$ if $y \prec_{{\revrelarank}_{\alpha}} x$.
\end{enumerate}

Now we describe the algorithm formally. We first define
\begin{align*}
T_{1} = \{\langle B, \prec \rangle \ | \ (B \subseteq A) \land (\prec \, \text{is a well ordering on }B) \}
\end{align*}
$T_{1}$ is a set, because if $\langle B, \prec \rangle \in T_{1}$,
then $\langle B, \prec \rangle \in \mathcal{P}(A) \times \mathcal{P}(A \times A)$.
By Axiom of Choice, there exists a function ${\getwell} \subseteq T_{1}$ with dom {\getwell} = dom$\ T_{1}$ = $\mathcal{P}(A)$.
That is, ${\getwell}(B)$ is a well ordering on $B \subseteq A$.
{\getwell} is one-to-one too.

Next we enumerate ${\revrelarank}$-decompositions of $A$.
Let $\gamma_{2}(f, y)$ be the formula:
	\begin{enumerate}[(i)]
		\item
		If $f$ is a function with domain an ordinal $\alpha \in {\rangerelarank}$,
		$y = {\getwell}({\revrelarank}_{\alpha}) \ \cup \ ((\bigcup {\revrelarank} \lsem \alpha \rsem) \times {\revrelarank}_{\alpha})$.

		\item
		otherwise, $y = \varnothing$.
	\end{enumerate}

To be mentioned again, ${\revrelarank} \lsem \alpha \rsem = \{ {\revrelarank}_{\beta} \ | \ \beta \in \alpha \}$.
If $\gamma_{2}(f, y_1)$ and $\gamma_{2}(f, y_2)$,
it is obvious that $y_1 = y_2$.
Then transfinite recursion theorem schema on well-ordered structures
gives us a unique function {\extendwell} with domain ${\rangerelarank}$
such that $\gamma_{2}({\extendwell} \upharpoonright \text{seg} \ \alpha, {\extendwell}(\alpha))$
for all $\alpha \in {\rangerelarank}$.
Because seg $\alpha = \alpha$, we get $\gamma_{2}({\extendwell} \upharpoonright \alpha, {\extendwell}(\alpha))$.

We claim that:
\begin{lemma}\label{is_well}
$W = \bigcup {\range}{\extendwell}$ is a well ordering on $A$ extended from $R$.
\end{lemma}
\begin{proof}
Suppose $x \in {\revrelarank}_{\alpha}, y \in {\revrelarank}_{\beta} \text{ and } z \in {\revrelarank}_{\theta}$
in which $\alpha, \beta, \theta \in {\rangerelarank}$.
\begin{enumerate}
	\item
	\begin{align*}
	\langle x, y \rangle \in R \ &\Rightarrow \ \alpha \in \beta\\
	&\Rightarrow \ \langle x, y \rangle \in (\bigcup {\revrelarank} \lsem \beta \rsem) \times {\revrelarank}_{\beta}\\
	&\Rightarrow \ \langle x, y \rangle \in {\extendwell}(\beta)\\
	&\Rightarrow \ \langle x, y \rangle \in W
	\end{align*}	
	Therefore $R \subseteq W$.
	
	\item
	There are three possible relations between $\alpha$ and $\beta$:
	\begin{enumerate}[(i)]
		\item
		$\alpha \in \beta$, then $x \neq y$ and $x\,W\,y$ according to the construction of $W$.
		
		\item
		$\alpha \ni \beta$, then $x \neq y$ and $y\,W\,x$.
		
		\item
		$\alpha = \beta$. Let $\prec_{{\revrelarank}_{\alpha}} \ = {\getwell}({\revrelarank}_{\alpha})$,
		then $x = y$, $x \prec_{{\revrelarank}_{\alpha}} y$, or $y \prec_{{\revrelarank}_{\alpha}} x$.
		This implies that $x = y$, $x\,W\,y$, or $y\,W\,x$.
	\end{enumerate}
		
	Furthermore suppose $x\,W\,y$ and $y\,W\,z$,
	then $\alpha \ \underline{\in} \ \beta \ \underline{\in} \ \theta$.
	If $\alpha \in \theta$, then $x\,W\,z$.
	Otherwise, $\alpha = \beta = \theta$.
	Let $\prec_{{\revrelarank}_{\alpha}} \ = {\getwell}({\revrelarank}_{\alpha})$, then $x \prec_{{\revrelarank}_{\alpha}} y$
	and $y \prec_{{\revrelarank}_{\alpha}} z$.
	Because $\prec_{{\revrelarank}_{\alpha}}$ is a well ordering, then $x \prec_{{\revrelarank}_{\alpha}} z$ and $x\,W\,z$.
	
	From the above, $W$ satisfies trichotomy on $A$ and is transitive,
	therefore $W$ is a linear ordering.
	
	\item
	Suppose $B$ is a nonempty subset of $A$,
	then ${\relarank} \lsem B \rsem$ is a nonempty set of ordinals by Axiom of Replacement.
	Such a set has a least element $\sigma$.
	Let $C = B \cap {\revrelarank}_{\sigma}$ and $\prec_{{\revrelarank}_{\sigma}} \ = {\getwell}({\revrelarank}_{\sigma})$.
	$C$ is a nonempty subset of ${\revrelarank}_{\sigma}$, so it has a least element $\hat{t}$ under $\prec_{{\revrelarank}_{\sigma}}$.
	For any $x$ in $B$ other than $\hat{t}$, either $\sigma \in \alpha$ or $\sigma = \alpha$.
	In both cases, $\hat{t}\,W\,x$ and $\hat{t}$ is indeed the least element of $B$.
\end{enumerate}
\end{proof}

Finally we conclude that an arbitrary well-founded or partially well-ordered structure can be extended to a well-ordered structure.

\section{Linear Extension Coincides Well Extension?}\label{sec:linear_coincide_well}
As mentioned earlier, any partial ordering can be linearly extended by Order-Extension Principle~\cite{Szpilrajn30}.
Is it possible that $\langle A, R \rangle$ can be always extended to a well-ordered structure?
Here is the result:
\begin{theorem}\label{every_linear_extension}
	A partially ordered structure $\langle A, R \rangle$ is partially well-ordered with no infinite totally unordered subset
	under $R$ if and only if every linear extension of $\langle A, R \rangle$ is well-ordered.
\end{theorem}
\begin{proof}
	Let $\langle A, L \rangle$ be an arbitrary linear extension of $\langle A, R \rangle$,
	and $<$ be the normal ordering on the set of natural numbers $\upomega$.
	\leavevmode
	\begin{enumerate}
		\item
		The "only if" part.
		Suppose $\langle A, L \rangle$ is not well-ordered,
		then there is an infinite sequence $s = \langle x_{n} : n \in \upomega \rangle$ in $A$
		(a function $f :\upomega \rightarrow A$)
		for which $x_{n^{+}} \, L \, x_{n}$ for all $n \in \upomega$.
		\begin{enumerate}[(i)]
			\item
			Clearly $A$ is an infinite set.			
			And elements in $s$ are distinct and {\range}$s$ is infinite.
			Otherwise there exists $x \in A$ such that $\langle x, x_{i_{1}} \cdots, x_{i_{k}}, x \rangle$ is a sub-sequence of $s$
			, which contradicts the fact that $L$ is irreflective.
			
			\item
			Let
			\begin{align*}
			T_{2} = \{ S_{\alpha} = {\revrelarank}_{\alpha} \cap {\range}s \ | \ (\alpha \in {\rangerelarank}) \land (S_{\alpha} \neq \emptyset) \}
			\end{align*}
			
			$T_{2}$ is a partition of {\range}$s$.
			By Axiom of Choice, there is a choice function $\getchoicea$ defined on $T_{2}$ such that $\getchoicea(\alpha) \in S_{\alpha}$.
			
			Let $e$ be an extraneous object not belonging to ${\range}s$.
			We define a function ${\getleast} : {\range}s \rightarrow {\range}s \cup \{e\}$ such that for any $B \subseteq {\range}s$:
			\begin{align*}
				{\getleast}(B) = \begin{cases} {\getchoicea}(\text{the least ordinal of }{\relarank} \lsem B \rsem), & \mbox{if } B \neq \emptyset \\ e, & \mbox{if } B = \emptyset \end{cases}
			\end{align*}
			
			${\getleast}$ does exist, because if $B$ is nonempty
			then ${\relarank} \lsem B \rsem$ is a nonempty set of ordinals by Axiom of Replacement.
			Such a set does have a least ordinal.
				
			\item
			Then we define by recursion a function ${\increasingsequence}$ from $\upomega$ into ${\range}s \cup \{e\}$:
			\begin{align*}
				{\increasingsequence}(0) &= {\getleast}({\range}s)\\
				{\increasingsequence}(n^{+}) &= {\getleast}(\{ x \ | \ (x \in {\range}s) \land (x \, L \, {\increasingsequence}(n)) \})
			\end{align*}
			
			${\increasingsequence}(n^{+}) \in {\range}s$ for each $n \in \upomega$
			because the set $\{ x \ | \ (x \in {\range}s) \land (x \, L \, {\increasingsequence}(n)) \}$ will always be infinite.
			Therefore ${\increasingsequence}$ is an infinite sub-sequence of $s$ and ${\relarank}({\increasingsequence}(n)) \ \underline{\in} \ {\relarank}({\increasingsequence}(n^{+}))$ for each $n \in \upomega$.
			
			\item
			Now we prove that ${\range}{\increasingsequence}$ is an infinite totally unordered subset of $A$.
			For two distinct $j, k \in \upomega$, let $j < k$ without loss of generality.
			Because ${\increasingsequence}(k) \, L \, {\increasingsequence}(j)$,
			either both ${\increasingsequence}(k)$ and ${\increasingsequence}(j)$ are incomparable,
			or ${\increasingsequence}(k) \, R \, {\increasingsequence}(j)$ as $L$ is the linear extension of $R$.
			The latter is impossible since ${\relarank}({\increasingsequence}(j)) \ \underline{\in} \ {\relarank}({\increasingsequence}(k))$.
		\end{enumerate}
		The above contradiction implies that $\langle A, L \rangle$ must be a well-ordered structure.

		\item
		The "if" part.
		\begin{enumerate}[(i)]
			\item
			$R$ is well-founded. Otherwise, $\langle A, R \rangle$ must have a descending chain
			$s = \langle x_{n} : n \in \upomega \rangle$ in $A$
			for which $x_{n^{+}} \, R \, x_{n}$.
			Because $L$ is the linear extension of $R$, $s$ also satisfies that $x_{n^{+}} \, L \, x_{n}$ for all $n \in \upomega$.
			Then $\langle A, L \rangle$ has a descending chain, and it could not be well-ordered.
			
			\item
			$A$ has no infinite totally unordered subsets under $R$.
			Otherwise, $A$ must have a countably infinite totally unordered subset $D$ under $R$.
			Let $f$ be the one-to-one function from $D$ onto the set of integers $\mathbb{Z}$,
			and $<_{\mathbb{Z}}$ be the normal ordering on $\mathbb{Z}$.
			We induce a linear ordering $<_{D}$ on $D$~\cite{Enderton77} by:
			\begin{align*}
			x <_{D} y \Leftrightarrow f(x) <_{\mathbb{Z}} f(y)
			\end{align*}
			$<_{D} \cup \, R$ is a partial ordering on $A$,
			since $<_{D}$ is a partial ordering disjointing with $R$.
			Then by Order-Extension Principle~\cite{Szpilrajn30}
			$<_{D} \cup \, R$ can be linearly extended to $L^\prime$, which is evidently one linear extension of $R$.
			$L^\prime$ is however not a well ordering, otherwise $<_{D}$ will be a well ordering on $D$ which is obviously false.
		\end{enumerate}
	\end{enumerate}
\end{proof}

The "if" part of Theorem \ref{every_linear_extension} is an existence proof.
In the following we take a countably infinite binary tree as an example
to illustrate how to construct a non-well linear extension.
The idea is to linearly extend such a tree by making the left subtree of each node \textit{greater} than its right subtree.

To be more precise, let $<$ be the normal ordering on the set of natural numbers $\upomega$,
and $R_{1} = \{\langle n, 2 \times n + 1 \rangle,  \langle n, 2 \times n + 2 \rangle \ | \ n \in \upomega \}$.
$\langle \upomega, R_{1} \rangle$ is a well-founded structure since $R_{1} \subseteq \ <$.
Let $R$ be the transitive extension of $R_{1}$,
then the partially well-ordered structure $\langle \upomega, R \rangle$
is the above mentioned countably infinite binary tree with the following properties:
\leavevmode
\begin{enumerate}[(a)]		
	\item
	$x\,R\,y \Rightarrow \exists z_{1}, z_{2}, \cdots, z_{n} \in \upomega \land x\,R_{1}\,z_{1}\,R_{1}\,z_{2}\,R_{1}\,\cdots\,R_{1}\,z_{n}\,R_{1}\,y$
	
	\item
	$R \subseteq \, <$
	
	\item
	${\rangerelarank} = {\range}{\relarank} = \upomega$
	
	\item
	${\revrelarank}_{n} = \{2^n - 1, 2^n, \cdots, 2^{n+1} - 2\}$ for all $n \in \upomega$,
	and card ${\revrelarank}_{n} = 2^n \in \upomega$.
	
	\item
	$\langle \upomega, R \rangle$ has infinite totally unordered subsets under $R$.
	Actually, $\{2^{n + 2} - 3 \ | \ n \in \upomega\}$ is one.
\end{enumerate}

We define the following function for each "node" to get its \textit{descendants}:
\begin{align*}
{\getdescendant} = \{\langle x, B \rangle \ | \ (x \in \upomega) \land (B \subseteq \upomega) \land (y \in B \Leftrightarrow x\,R\,y)\}
\end{align*}
{\getdescendant} is a function from $\upomega$ into $\mathcal{P}(\upomega)$,
because it is a subset of $\upomega \times \mathcal{P}(\upomega)$ and is single rooted.

Let $\gamma_{3}(f, y)$ be the formula:
\begin{enumerate}[(i)]
	\item
	$f$ is a function with domain a natural number $n \in \upomega$.
	Denote ${\revrelarank}_{n}$ as $\{x_{1}, x_{2}, \cdots, x_{2^n}\}$
	for which $x_{1} < x_{2} < \cdots < x_{2^n}$(they are totally unordered under $R$).
	Then $y = \bigcup\limits_{1 \leq i < j \leq 2^n}({\getdescendant}(x_{j}) \times {\getdescendant}(x_{i}))$
	
	\item
	otherwise, $y = \varnothing$.
\end{enumerate}

Transfinite recursion theorem schema gives us a unique function {\extendlinear}
with domain $\upomega$ such that $\gamma_{3}({\extendlinear} \upharpoonright \text{seg} \ n, {\extendlinear}(n))$
for all $n \in \upomega$.
That is, $\gamma_{3}({\extendlinear} \upharpoonright n, {\extendlinear}(n))$.
Then $L = (\bigcup {\range}{\extendlinear}) \cup R$ is a linear extension of $R$.
The proof is straightforward, and we omit the details here.
Let $s = \langle x_{n} = 2^{n + 2} - 3 : n \in \upomega \rangle$.
It is easy to verify that $x_{n^{+}}\,L\,x_{n}$ for all $n \in \upomega$.
Therefore $s$ is a descending chain and $L$ cannot be a well ordering on $\upomega$.

\end{document}